\documentclass[letterpaper, 10 pt, conference]{ieeeconf}  

\IEEEoverridecommandlockouts                              
\usepackage{amsmath,amsfonts,amssymb,color}
\usepackage{graphicx}
\usepackage{psfrag}
\usepackage{algorithm}
\usepackage{algpseudocode}
\usepackage{array}
\usepackage{cite}
\usepackage{footmisc}
\usepackage{mathtools}
\usepackage{enumerate}
\usepackage{soul,color}
\usepackage{caption}
\usepackage{subfig}
\DeclarePairedDelimiter{\ceil}{\lceil}{\rceil}

\newtheorem{theorem}{Theorem}

\newtheorem{lemma}{Lemma}
\newtheorem{remark}{Remark}
\newtheorem{definition}{Definition}
\newtheorem{proposition}{Proposition}
\newtheorem{example}{Example}
\newtheorem{assumption}{Assumption}

\renewcommand{\theenumi}{(\alph{enumi}}
\newtheorem{Remark}{Remark}

\DeclareMathOperator{\diag}{diag}

\DeclarePairedDelimiter\floor{\lfloor}{\rfloor}

\makeatletter
\newcommand{\distas}[1]{\mathbin{\overset{#1}{\kern\z@\sim}}}%
\newsavebox{\mybox}\newsavebox{\mysim}
\newcommand{\distras}[1]{%
  \savebox{\mybox}{\hbox{\kern3pt$\scriptstyle#1$\kern3pt}}%
  \savebox{\mysim}{\hbox{$\sim$}}%
  \mathbin{\overset{#1}{\kern\z@\resizebox{\wd\mybox}{\ht\mysim}{$\sim$}}}%
}
\makeatother

\title{\LARGE \bf
Learning Linearized Models from Nonlinear Systems with Finite Data

}

\author{~Lei~Xin, George Chiu, Shreyas Sundaram 
\thanks{This work was partially supported by a contract from Saab Inc., under DARPA Award N65236-23-C-8012. Lei Xin and Shreyas Sundaram are with the Elmore Family School of Electrical and Computer Engineering, Purdue University. George Chiu is with the School of Mechanical Engineering, Purdue University. E-mails: {\tt\{lxin, gchiu, sundara2\}@purdue.edu}.
}
}

\begin{document}

\maketitle
\thispagestyle{empty}
\pagestyle{empty}

\begin{abstract}
Identifying a linear system model from data has wide applications in control theory. The existing work on finite sample analysis for linear system identification typically uses data from a single system trajectory under i.i.d  random inputs, and assumes that the underlying dynamics is truly linear. In contrast, we consider the problem of identifying a linearized model when the true underlying dynamics is nonlinear. We provide a multiple trajectories-based deterministic data acquisition algorithm followed by a regularized least squares algorithm, and provide a finite sample error bound on the learned linearized dynamics. Our error bound demonstrates a trade-off between the error due to nonlinearity and the error due to noise, and shows that one can learn the linearized dynamics with arbitrarily small error given sufficiently many samples. We validate our results through experiments, where we also show the potential insufficiency of linear system identification using a single trajectory with i.i.d random inputs, when nonlinearity does exist.
\end{abstract}


\section{Introduction} \label{sec: introduction}
Learning good predictive models from data has wide applications, including in economics and machine learning \cite{athey2018impact, mitchell2007machine}. The problem of system identification is to learn a mathematical model of a dynamical system from data. System identification is an important problem in control theory since a good model can facilitate model-based control design \cite{ljung1999system}. Although physical systems are typically nonlinear, linear models are frequently used in practice due to their simplicity \cite{rugh1996linear}, and their ability to approximate nonlinear systems around a given reference point. Consequently, it is of interest to understand identification of appropriate linear models from data generated by nonlinear systems.

Classically, theories for system identification typically focus on asymptotic aspects \cite{bauer1999consistency,jansson1998consistency}. In recent years, however, finite sample analysis for system identification has been studied extensively. For linear system identification, existing works are either multiple trajectories-based or single trajectory-based. The multiple trajectories setup \cite{dean2019sample, fattahi2018data, zheng2020non, xin2022learning} requires the user to restart the system multiple times, but has a major advantage in its ability to handle unstable systems. In contrast, the single trajectory setup \cite{simchowitz2018learning,oymak2019non,simchowitz2019learning,sarkar2019nonparametric,faradonbeh2018finite,sarkar2019near} performs system identification using data from a single experiment, i.e., the system does not need to reset,  but has potential risks if the system is unstable.  We note that when it comes to linear system identification, almost all existing works that have finite sample guarantees assume that the underlying system is truly linear, except for \cite{sarker2023accurate}. Furthermore, Gaussian random inputs are typically applied to ensure persistent excitation. 

The study on nonlinear system identification is less well-understood, in general, as compared to the case for linear system identification. Recent works on finite sample analysis for nonlinear system identification include \cite{sattar2022non, mania2020active, foster2020learning}. It is worth noting that to obtain finite sample guarantees, the existing works on nonlinear system identification typically require that a certain model structure to be known in advance. However, when the specific model structure is unknown, a reasonable alternative goal is to learn a linearized model from the nonlinear system, due to the well-studied techniques on linear system control as discussed above.

There is a branch of recent research that focuses on learning a linear system representation that completely captures the behaviours of a nonlinear system using the Koopman Operator \cite{mauroy2016linear}. This approach typically requires carefully selected basis functions (e.g., leveraging neural networks \cite{hao2022deep}), and the analysis focuses on the noiseless setting.  In contrast, our focus in this work is to learn a linearized system model, in the sense that the linear model captures the linear part of the nonlinear system after Taylor expansion, and to provide finite sample guarantees when the system has noise.

Most relevant to our work is the recent paper \cite{sarker2023accurate}, which provides a finite sample error bound for learning linear models from systems that have unmodeled dynamics that could capture nonlinearities, using a single system trajectory. However, the method proposed in \cite{sarker2023accurate} assumes the system is ``well-behaved'' by requiring the unmodeled dynamics/nonlinear terms to be (globally) Lipschitz \cite{cobzacs2019lipschitz}. The method also requires the system to satisfy certain additional properties to ensure consistent estimation, supposing the inputs are carefully chosen. In contrast, we show in this work that one can learn a linearized system model from a nonlinear system with arbitrarily small error without the Lipschitzness assumption, given sufficiently many short trajectories, supposing that one has control over the initial conditions of the experiments. 

In summary, our contributions are as follows. 
\begin{itemize}
  \item We provide a deterministic, multiple trajectories-based data acquisition algorithm that ensures persistent excitation under the constraint of being close to the reference point. Using this algorithm followed by a regularized least squares estimation algorithm, we provide a finite sample error bound of the identified linearized dynamics of a nonlinear system.
  \item Our bound shows that one can learn the linearized dynamics with arbitrarily small error, given sufficiently many experiments in the multiple trajectories setup, and demonstrates a trade-off between the error due to noise and the error due to nonlinearity. The bound further characterizes the benefits of using regularization. When the system is perfectly linear, we show a learning rate that matches the existing results on learning perfectly linear systems using random inputs.
  \item We provide numerical experiments to validate our results and insights, and show the potential insufficiency of linear system identification using random inputs from a single trajectory when nonlinearity does exist.
  \end{itemize}
  
Our paper is organized as follows. Section \ref{notation} introduces relevant mathematical notation. Section \ref{sec:algorithm} introduces the system identification problem and the algorithms we use. In Section \ref{analysis}, we present our theoretical results. We present numerical examples in in Section \ref{exp} to validate our results, and conclude in Section \ref{sec: conclusion}. 

\section{Notation} \label{notation}
Vectors are taken to be column vectors unless indicated otherwise. Let $\mathbb{R}$ and $\mathbb{Z}$ denote the set of real numbers and integers, respectively. Let $\lambda_{max}(\cdot)$ and $\lambda_{min}(\cdot)$ be the largest and the smallest eigenvalue in magnitude, respectively, of a given matrix. For a given matrix $A$, we use $A'$ to denote its conjugate transpose. We use $\|A\|$, $\|A\|_{1}$ and $\|A\|_{F}$ to denote the spectral norm, $1$-norm, and Frobenius norm, respectively, of matrix $A$. We use $I_{n}$ to denote the identity matrix with dimension $n$. We use the symbol $\bmod$ to denote the modulo operation. The union of sets is denoted as $\cup$. The open $l_{1}$ ball in $d$-dimensional space with center at $x_{0}$ and radius $r$ is
denoted by $\mathcal{B}_{d}(x_{0},r)\triangleq \{x\in \mathbb{R}^{d}:\|x-x_{0}\|_{1}< r\}$. We denote $e_{i}^{d}$ as a $d$-dimensional vector with the $i$-th component equal to 1 and all other components equal to 0. The symbols $\floor{\cdot}$ and $\ceil{\cdot}$ are used to denote the floor and ceiling functions, respectively. We use $\textbf{0}$ to denote a zero vector with dimension that is clear from the context. The symbol $\sigma(\cdot)$ is used to denote the sigma field generated by the corresponding random vectors. The symbol $\mathcal{S}^{n-1}$ is used to denote the unit sphere in $n$-dimensional space.


\section{Problem Formulation and System Identification Algorithm} \label{sec:algorithm}
Consider the following discrete time nonlinear time invariant system
\begin{equation} \label{nonlinear}
\begin{aligned}
x_{k+1}=f(z_{k})+w_{k},
\end{aligned}
\end{equation}
where $f: \mathbb{R}^{n+p} \to \mathbb{R}^{n}$, $z_{k}=\begin{bmatrix} x_{k}'&u_{k}' \end{bmatrix}^{'}\in \mathbb{R}^{n+p}$, $x_{k}\in\mathbb{R}^{n}$, $u_{k}\in \mathbb{R}^{p}$, and $w_{k}\in \mathbb{R}^{n}$. Here, $x_{k}, u_{k}$ and $w_{k}$ are the state, input, and process noise, respectively. The noise terms $w_{k}$ are  assumed to be independent sub-Gaussian random vectors with parameter $\sigma_{w}^2$, where the definition is given below \cite{rivasplata2012subgaussian}.
\begin{definition}
A real-valued random variable $w$ is called sub-Gaussian with parameter $\sigma^2$ if  we have
\begin{equation*}
\begin{aligned}
&\forall \alpha\in \mathbb{R}, \mathbb{E}[\exp(\alpha{w})]\leq \exp(\frac{\alpha^2 \sigma^2}{2}).\\
\end{aligned}
\end{equation*}
A random vector $x\in \mathbb{R}^n$ is called $\sigma^2$ sub-Gaussian if for all unit vectors $v\in \mathcal{S}^{n-1}$ the random variable $v'x$ is $\sigma^2$ sub-Gaussian.
\end{definition}

Note that sub-Gaussian distributions are commonly used to model noise processes \cite{sarkar2019near}. In particular, every (zero-mean) Gaussian random vector is sub-Gaussian.

Assume that for each component function of $f$, all second order partial derivatives exist and are continuous on $\mathbb{R}^{n+p}$. From Taylor's theorem \cite{courant1965introduction}, system \eqref{nonlinear} using reference point $z_{k}=\textbf{0}$ can be rewritten as
\begin{equation} \label{system}
\begin{aligned} 
x_{k+1}=A x_{k}+Bu_{k}+o+w_{k}+r_{k}, \\
\end{aligned}
\end{equation}
where $A\in \mathbb{R}^{n\times n},B\in \mathbb{R}^{n\times p}$, are system matrices that capture the linear part of $f(z_{k})$, $o=f(\textbf{0})\in \mathbb{R}^{n}$, and $r_{k}=h(z_{k})\in \mathbb{R}^{n}$ is a 
remainder vector that contains higher order terms that are state/input dependent, where $h: \mathbb{R}^{n+p} \to \mathbb{R}^{n}$. The above model is less studied in the literature on finite sample analysis for system identification, and we will consider this model in the sequel. Note that we assume $o$ is possibly non-zero to capture scenarios where the equilibrium points of the system are unknown.  When the system is perfectly linear, we have $o= r_{k}=\textbf{0}$, which is the commonly used model in the literature. In this paper, we assume that both the state $x_{k}$ and input $u_{k}$ can be perfectly measured.
Suppose that we can restart the system multiple times from an arbitrary initial state $x_{0}$ using arbitrary input $u_{0}$, and obtain multiple length 1 trajectories (i.e., state-input pairs obtained by running the system for a single time step, as will be explained next). Using superscript to denote the trajectory index, we denote the set of samples we have as $\{(x^{i}_{1}, x^{i}_{0},u^{i}_{0}):1 \leq i \leq N\}$.  Our goal is to learn the linear approximation system matrices  $\Theta \triangleq \begin{bmatrix}
A&B&o \end{bmatrix}\in \mathbb{R}^{n\times(n+p+1)}$ in system \eqref{system} from the set of samples available to us.

Our result will leverage the following mild assumption on the remainder vector $r_{k}=h(z_{k})$ in system \eqref{system}.
\begin{assumption} \label{ass:remainder}
Let $r_{i,k}$ denote the $i$-th component of $r_{k}$. There exist $c>0$ and $\beta=\beta(c)$ such that $|r_{i,k}|\leq \beta \|z_{k}\|_{1}^2$ for all $i \in \{1,\ldots, n\}$ and all $z_{k}\in \mathcal{B}_{n+p}(\textbf{0},c)$.  
\end{assumption}

\begin{Remark}
The above assumption is, in fact, a direct result of assuming that each component function of the original nonlinear dynamics $f$ has all second order partial derivatives being continuous on $\mathbb{R}^{n+p}$, due to Taylor's theorem for multivariable functions from \cite[Corollary~1]{folland2005higher}. Intuitively, this assumption says that the higher order terms are dominated by the second order terms, if the arguments of the function are sufficiently close to the origin. Note that it does not require the function $h$ to be Lipschitz (which is the assumption used in \cite{sarker2023accurate}). As an example, consider a scalar system with the dynamics given by $f(z_{k})=x_{k}+u_{k}+x_{k}^2+x_{k}^3$. Here $r_{k}=x_{k}^2+x_{k}^3$ satisfies Assumption \ref{ass:remainder} for $c=1$ and $\beta=2$ since $|x_{k}^2+x_{k}^3|\leq |x_{k}^2|+|x_{k}^3|\leq 2|x_{k}|^2\leq 2\|z_{k}\|_{1}^2$ for all $z_{k}\in \mathcal{B}_{2}(\textbf{0},1)$, but the corresponding function $h$ is not Lipschitz. In general, a larger $c$ may lead to a larger $\beta$. 
\end{Remark}

Let $q>0$ be a (small) design parameter that constrains the magnitude of the initial conditions $z_{0}$,  and $N$  be the number of experiments to perform. We deploy a data collection scheme specified in Algorithm \ref{algo1}. 
\begin{algorithm}[H]
\caption{Data Acquisition} \label{algo1}
\textbf{Input} Norm constraint parameter $q>0$, number of experiments $N>0$
\begin{algorithmic}[1]
\State Initialize $s_{1}=1$
\For {$i=1,\ldots, N$}
\If{$i\bmod{(n+p)}\neq 0$}
\State Set $z^{i}_{0}=\left[\begin{smallmatrix} x_{0}^{i'}&u_{0}^{i'} \end{smallmatrix}\right]^{'}=s_{i}\times q e_{i\bmod{(n+p)}}^{n+p}$
\State Collect $x^{i}_{1}$, where $x^{i}_{1}=Ax^{i}_{0}+Bu^{i}_{0}+w^{i}_{0}+o+r^{i}_{0}$
\State Set $s_{i+1}=s_{i}$
\Else
\State Set $z^{i}_{0}=\left[\begin{smallmatrix} x_{0}^{i'}&u_{0}^{i'} \end{smallmatrix}\right]^{'}=s_{i}\times q e_{n+p}^{n+p}$
\State Collect $x^{i}_{1}$, where $x^{i}_{1}=Ax^{i}_{0}+Bu^{i}_{0}+w^{i}_{0}+o+r^{i}_{0}$
\State Set $s_{i+1}=-s_{i}$
\EndIf
\EndFor
\State Output  $\{(x^{i}_{1}, x^{i}_{0}, u^{i}_{0}):1 \leq i \leq N\}$
\end{algorithmic}
\end{algorithm}
\begin{Remark}
Intuitively, we want the initial conditions to stay as close to the reference point (in this case, the origin) as possible, to avoid excessive bias from the higher order terms. Hence, the reason of using of multiple length 1 trajectories is to prevent the noise from driving the system too far from the reference point, and amplifying the effects from $r_{k}$. The key idea of Algorithm \ref{algo1} is to ensure persistent excitation (i.e., the smallest eigenvalue of the sample covariance matrix becomes larger as one gets more data), subject to the constraint on bounded distance to the origin (specified by $q$). Later on in our theoretical result, we will demonstrate how $q$ will affect the finite sample estimation error bound for learning $\Theta$.
\end{Remark}

We establish some definitions now. Define the batch matrices 
\begin{equation} 
\begin{aligned} 
&X=\begin{bmatrix} x^{1}_{1}&x^{2}_{1}&\cdots& x^{N}_{1}\end{bmatrix}\in \mathbb{R}^{n\times N} \\
&W=\begin{bmatrix} w^{1}_{0}&w^{2}_{0}&\cdots& w^{N}_{0}\end{bmatrix}\in \mathbb{R}^{n \times N}\\
&R=\begin{bmatrix} r^{1}_{0}&r^{2}_{0}&\cdots& r^{N}_{0}\end{bmatrix}\in \mathbb{R}^{n \times N}.
\end{aligned}
\end{equation}

Let $\hat{z}^{i}_{0}=\begin{bmatrix} z^{i'}_{0}&1 \end{bmatrix}^{'}\in \mathbb{R}^{n+p+1}$. Define the regressor matrix
\begin{equation} 
\begin{aligned} 
&Z=\begin{bmatrix} \hat{z}^{1}_{0}&\hat{z}^{2}_{0}&\cdots &\hat{z}^{N}_{0}\end{bmatrix}\in \mathbb{R}^{(n+p+1)\times N}.
\end{aligned}
\end{equation}

We have the following relationship
\begin{equation}
\begin{aligned} 
X=\Theta Z+W+R. \\
\end{aligned}
\end{equation}

To learn the linear model $\Theta$, we would like to solve the following regularized least squares problem
\begin{equation*}
\begin{aligned}
  \mathop{\min}_{\tilde{\Theta}\in \mathbb{R}^{n\times (n+p+1)}} \{\|X-\tilde{\Theta}Z\|^{2}_{F}+\lambda \|\tilde{\Theta}\|^2_{F}\},
\end{aligned}
\end{equation*}
where $\lambda\geq 0$ is a regularization parameter. The closed-form solution of the above problem is given by
\begin{equation} 
\begin{aligned}
\hat{\Theta}=XZ^{'}(ZZ'+\lambda I_{n+p+1})^{-1},
\end{aligned}
\end{equation}
under the invertibility assumption \cite{hoerl1970ridge}. The estimation error is then given by
\begin{equation} 
\begin{aligned}
\|\hat{\Theta}-\Theta\|&=\|-\lambda\Theta(ZZ'+\lambda I_{n+p+1})^{-1}\\
&+WZ'(ZZ'+\lambda I_{n+p+1})^{-1}\\
&+RZ'(ZZ'+\lambda I_{n+p+1})^{-1}\| .\label{error}
\end{aligned}
\end{equation}

For the ease of reference, the above steps are encapsulated in Algorithm \ref{algo2}.

\begin{algorithm}[H]
\caption{System Identification Using Multiple Length $1$ Trajectories} \label{algo2}
\textbf{Input} Dataset $\{(x^{i}_{1}, x^{i}_{0}, u^{i}_{0}):1 \leq i \leq N\}$, regularization parameter $\lambda\geq 0$
\begin{algorithmic}[1]  
\State Construct the matrices $X,Z$. Compute $\hat{\Theta}=XZ'(ZZ'+\lambda I_{n+p+1})^{-1}$.
\State Extract the estimated system matrices $A,B,o$ from the estimate $\hat{\Theta}=\begin{bmatrix}\hat{A}&\hat{B}&\hat{o}\end{bmatrix}$.
\end{algorithmic}
\end{algorithm}

In the next section, we will provide a finite sample bound of the system identification error \eqref{error} using Algorithm \ref{algo1} and Algorithm \ref{algo2}. The bound explicitly characterizes how the error depends on $N$, $q$, $\sigma_{w}$, $\lambda$, and other system parameters, and will provide guidance on selecting $q,\lambda$.
 

\section{Theoretical Analysis} \label{analysis}
To upper bound the system identification error in \eqref{error} with high probability, we bound the terms $\|-\lambda\Theta(ZZ'+\lambda I_{n+p+1})^{-1}\|, \|WZ'(ZZ'+\lambda I_{n+p+1})^{-1/2}\|, \|(ZZ'+\lambda I_{n+p+1})^{-1/2}\|, \|RZ'(ZZ'+\lambda I_{n+p+1})^{-1}\|$ separately. We provide some intermediate results first in Section \ref{intermediate}. Our main result is presented in Section \ref{main}.
\subsection{Intermediate results} \label{intermediate}
The following result shows the persistent excitation property of Algorithm \ref{algo1}. Note that the requirement on $N\geq 4(n+p)$ below is mainly used for numerical simplification.
\begin{lemma} 
Suppose that Algorithm \ref{algo1} is used to generate data. Let $N\geq 4(n+p)$. Then we have the following inequalities
\begin{equation*} 
\begin{aligned} 
&\lambda_{min}(ZZ')\geq N \min\{\frac{q^2}{2(n+p)},\frac{1}{2}\},\\
&\lambda_{max}(ZZ')\leq N \max\{\frac{2q^2}{n+p},2\}.
\end{aligned}
\end{equation*}

\label{lemma:PE}
\end{lemma}
\begin{proof}
To ease the notation, we write $e_{i}^{n+p}$ as $e_{i}$ for $i=1,\ldots,n+p$ in the sequel. We focus on the lower bound first. Denote $N_{1}=\floor{\frac{N}{2(n+p)}}\times {2(n+p)}$. Since the assumption $N\geq 4(n+p)$ implies $N_{1}\geq 4(n+p)$, we have
\begin{equation} 
\begin{aligned} 
ZZ'&=\sum_{i=1}^{N}\hat{z}^{i}_{0}\hat{z}^{i'}_{0}\succeq \sum_{i=1}^{N_{1}}\hat{z}^{i}_{0}\hat{z}^{i'}_{0}\\
&=\left(\sum_{i=1,1+(n+p),1+2(n+p),\ldots}^{N_{1}-(n+p)+1} \begin{bmatrix} s_{i}e_{1}q\\ 1\end{bmatrix}\begin{bmatrix} (s_{i}e_{1}q)' & 1\end{bmatrix}\right)\\
&+\left(\sum_{i=2,2+(n+p), 2+2(n+p),\ldots}^{N_{1}-(n+p)+2} \begin{bmatrix} s_{i}e_{2}q\\ 1\end{bmatrix}\begin{bmatrix} (s_{i}e_{2}q)' & 1\end{bmatrix}\right)\\
&+ \cdots\\
&+\left(\sum_{i=n+p,n+p+(n+p),\ldots}^{N_{1}} \begin{bmatrix} s_{i}e_{n+p}q\\ 1\end{bmatrix}\begin{bmatrix} (s_{i}e_{n+p}q)' & 1\end{bmatrix}\right)\\
&=\begin{bmatrix} M_{1}&M_{2}\\ M_{2}'&M_{3}\end{bmatrix},
\end{aligned}
\end{equation}
where $M_{1}\in \mathbb{R}^{(n+p)\times (n+p)}, M_{2}\in \mathbb{R}^{(n+p)\times 1}$, and $M_{3}\in \mathbb{R}^{1\times 1}$.

For the submatrix $M_{1}$, we have
\begin{equation}  \label{M1}
\begin{aligned} 
M_{1}&=\sum_{j=1}^{n+p}\sum_{i=1}^{\frac{N_{1}}{n+p}}e_{j}e_{j}'q^2=\sum_{j=1}^{n+p}\frac{N_{1}}{n+p}e_{j}e_{j}'q^2\\
&=\diag(\frac{N_{1}}{n+p}q^2,\cdots,\frac{N_{1}}{n+p}q^2),
\end{aligned}
\end{equation}
where we used the property that $s_{i}^2=1$ for all $i$, and the fact that $N_{1}\bmod{2(n+p)}=0$

For the submatrix $M_{2}$, we have
\begin{equation} \label{M2}
\begin{aligned} 
M_{2}&=\left(\sum_{i=1,1+(n+p),1+2(n+p),\ldots}^{N_{1}-(n+p)+1}s_{i}e_{1}q\right)+\cdots\\
&+\left(\sum_{i=n+p,n+p+(n+p),\ldots}^{N_{1}}s_{i}e_{n+p}q\right)\\
&=\textbf{0}+\textbf{0}+\ldots+\textbf{0}=\textbf{0},
\end{aligned}
\end{equation}
where we used the property that $s_{i}=1$ if $i\in \{j(n+p)+1,j(n+p)+2,\ldots, j(n+p)+(n+p)|j \text{ is even}\}$ and $s_{i}=-1$ if $i\in \{j(n+p)+1,j(n+p)+2,\ldots, j(n+p)+(n+p)|j \text{ is odd}\}$, and the fact that $N_{1}\bmod{2(n+p)}=0$, i.e., the number of positive terms is exactly the same as the number of negative terms for each summation.

Lastly, for the scalar matrix $M_{3}$, we have 
\begin{equation} \label{M3}
\begin{aligned} 
M_{3}&=\sum_{i=1}^{N_{1}}1^2=N_{1}.
\end{aligned}
\end{equation}
Combining \eqref{M1}-\eqref{M3}, we have 
\begin{equation} \label{touse1}
\begin{aligned} 
\lambda_{min}(ZZ')\geq \lambda_{min}\left(\begin{bmatrix} M_{1}&M_{2}\\ M_{2}'&M_{3}\end{bmatrix}\right)=\min\{\frac{N_{1}}{n+p}q^2,N_{1}\}. 
\end{aligned}
\end{equation}
Using the property $\floor{\frac{N}{c}}c\geq N-c$ for any $c>0$, we have 
\begin{equation}
\begin{aligned} 
N_{1}=\floor{\frac{N}{2(n+p)}}\times {2(n+p)}\geq N-2(n+p)\geq \frac{N}{2},
\end{aligned}
\end{equation}
where the second inequality is due to our assumption that $N\geq 4(n+p)$.

Hence, the above inequality in conjunction with \eqref{touse1} yields
\begin{equation} 
\begin{aligned} 
\lambda_{min}(ZZ')\geq N \min\{\frac{q^2}{2(n+p)},\frac{1}{2}\}, 
\end{aligned}
\end{equation}
which is of the desired form.

Next, we show the upper bound. Denoting $N_{2}=\ceil{\frac{N}{2(n+p)}}\times{2(n+p)}$, using $N\leq N_{2}$, we have
\begin{equation} 
\begin{aligned} 
ZZ'&=\sum_{i=1}^{N}\hat{z}^{i}_{0}\hat{z}^{i'}_{0}\preceq \sum_{i=1}^{N_{2}}\hat{z}^{i}_{0}\hat{z}^{i'}_{0},
\end{aligned}
\end{equation}
where $\hat{z}^{1}_{0},\hat{z}^{2}_{0},\ldots, \hat{z}^{N_{2}}_{0}$ are generated from Algorithm \ref{algo1} with input parameter $N_{2}$. Since $N_{2}\bmod{2(n+p)}=0$, we can follow a similar procedure as in the proof of the lower bound to obtain 
\begin{equation} 
\begin{aligned} 
\lambda_{max}(ZZ')&\leq \max\{\frac{N_{2}}{n+p}q^2,N_{2}\}\\
&\leq \max\{\frac{N+2(n+p)}{n+p}q^2,N+2(n+p)\}\\
&\leq N\max\{\frac{2q^2}{n+p},2\},
\end{aligned}
\end{equation}
where the second inequality is due to the relationship $N_{2}\leq N+2(n+p)$, and the last inequality is due to the assumption that $N\geq 4(n+p)$.
\end{proof}

We will use the following lemma, which generalizes the upper bound for self-normalized martingales in \cite[Theorem~1]{abbasi2011improved}  to the multi-dimensional case. The proof can be found in \cite[Lemma~5]{xin2023learning}.

\begin{lemma} \label{martingale_bound_multi}
Let $\{\mathcal{F}_{t}\}_{t\geq 0}$ be a filtration. Let $\{{w}_{t}\}_{t\geq 1}$ be a  $\mathbb{R}^{n}$-valued stochastic process such that $w_{t}$ is $\mathcal{F}_{t}$-measurable, and $w_{t}$ is conditionally sub-Gaussian on $\mathcal{F}_{t-1}$ with parameter $R^2$. Let $\{z_{t}\}_{t\geq 1}$ be an $\mathbb{R}^{m}$-valued stochastic process such that $z_{t}$
is $\mathcal{F}_{t-1}$-measurable. Assume that $V$ is a $m\times m$ dimensional positive definite matrix. For all $t\geq 0$, define
\begin{equation*}
\begin{aligned}
&\bar{V}_{t}=V+\sum_{s=1}^{t}z_{s}z_{s}', S_{t}=\sum_{s=1}^{t}z_{s}w_{s}'.\\
\end{aligned}
\end{equation*}
Then, for any $\delta\in (0,1)$, and for all $t\geq0$,
\begin{equation*}
\begin{aligned}
&P(\|\bar{V}_{t}^{-\frac{1}{2}}S_{t}\|\leq\sqrt{\frac{32}{9}R^{2}(\log\frac{9^n}{\delta}+\frac{1}{2}\log\det(\bar{V}_{t}V^{-1})})\\
&\geq 1-\delta.\\
\end{aligned}
\end{equation*}
\end{lemma}

We have the following result that upper bounds the contribution from noise.
\begin{lemma} \label{bound noise}
Suppose that Algorithm \ref{algo1} is used to generate data. Let $N\geq 4(n+p)$ and $q\leq \sqrt{n+p}$. Then for any fixed $\delta \in (0,1)$,  we have with probability at least $1-\delta$
\begin{equation*}
\begin{aligned}
&\|WZ'(ZZ'+\lambda I_{n+p+1})^{-1/2}\|\\
&\leq 3 \sigma_{w} \sqrt{\log\frac{9^n}{\delta}+(n+p+1)\log(1+\frac{4(n+p)}{q^2+\zeta})},
\end{aligned}
\end{equation*}
where $\zeta=\frac{4\lambda(n+p)}{N}$.
\end{lemma}

\begin{proof}
Denoting $\bar{V}_{N}=\lambda I_{n+p+1}+ZZ'$, we have
\begin{equation*}
\begin{aligned}
&\|WZ'(ZZ'+\lambda I_{n+p+1})^{-1/2}\|=\|\bar{V}_{N}^{-1/2}ZW'\|.
\end{aligned}
\end{equation*}
Let $\hat{V}_{N}=(\lambda+\frac{Nq^2}{2(n+p)})I_{n+p+1}$. When $N\geq 4(n+p)$ and $q\leq \sqrt{n+p}$, we can apply Lemma \ref{lemma:PE} to get $\bar{V}_{N}\succeq \hat{V}_{N}$. Since $\bar{V}_{N} \succeq \hat{V}_{N} \Rightarrow  2\bar{V}_{N} \succeq \bar{V}_{N}+\hat{V}_{N}\Rightarrow \bar{V}_{N}^{-1} \preceq 2 (\bar{V}_{N}+\hat{V}_{N})^{-1}$, we can write 
\begin{equation*}
\begin{aligned}
&\|\bar{V}_{N}^{-1/2}ZW'\|\leq \sqrt{2}\|(\bar{V}_{N}+\hat{V}_{N})^{-1/2}ZW'\|\\
&=\sqrt{2}\|(\hat{V}_{N}+\lambda I_{n+p+1})+\sum_{i=1}^{N}\hat{z}^{i}_{0}\hat{z}^{i'}_{0})^{-1/2}(\sum_{i=1}^{N}\hat{z}^{i}_{0}w^{i'}_{0})\|,
\end{aligned}
\end{equation*}
where the inequality is due to \cite[Lemma~10]{xin2023learning}.

Denote $V=\hat{V}_{N}+\lambda I_{n+p+1}$. Define the filtration $\{\mathcal{F}_{t}\}_{t\geq 0}$, where $\mathcal{F}_{t}=\sigma(\{\hat{z}_{0}^{i+1}\}_{i=0}^{t}\cup\{w_{0}^{j}\}_{j=1}^{t})$. Since the sequence of $\hat{z}^{i}_{0}$ generated by Algorithm \ref{algo1} is deterministic, and the noise terms are independent, for any fixed $\delta \in (0,1)$, we can apply Lemma \ref{martingale_bound_multi} to obtain with probability at least $1-\delta$
\begin{equation*}
\begin{aligned}
&\sqrt{2}\|(\bar{V}_{N}+\hat{V}_{N})^{-1/2}ZW'\|\\
&\leq 3 \sigma_{w} \sqrt{\log\frac{9^n}{\delta}+\frac{1}{2}\log\det((V+ZZ')V^{-1})}.
\end{aligned}
\end{equation*}
When $q\leq \sqrt{n+p}$, we can apply the upper bound in Lemma \ref{lemma:PE} to obtain
\begin{equation*}
\begin{aligned}
det((V+ZZ')V^{-1})&=\frac{det(V+ZZ')}{det(V)}\\
&\leq \frac{(2\lambda+\frac{Nq^2}{2(n+p)}+\|ZZ'\|)^{n+p+1}}{(2\lambda+\frac{Nq^2}{2(n+p)})^{n+p+1}}\\
&\leq (1+\frac{2N}{2\lambda+\frac{Nq^2}{2(n+p)}})^{n+p+1}\\
&= (1+\frac{4(n+p)}{q^2+\zeta})^{n+p+1},
\end{aligned}
\end{equation*}
where we used the fact that the determinant is the product of eigenvalues. The result then follows.
\end{proof}

Next, we bound the contribution from the higher order terms.
\begin{lemma} \label{bound nonlinearity}
Suppose that Algorithm \ref{algo1} is used to generate data.  Let $N\geq 4(n+p)$ and $q\leq \sqrt{n+p}$. Fix constants $c$ and $\beta$ that satisfy Assumption \ref{ass:remainder}, and denote $\gamma=\frac{\lambda(n+p)}{Nq^2}$. Then if $q< c$, we have 
\begin{equation} 
\begin{aligned}
&\|RZ'(ZZ'+\lambda I_{n+p+1})^{-1}\|\\
&\leq \sqrt{\frac{2\beta^2(n^2+np)}{1+\gamma}}q+\frac{2 (n+p)\sqrt{\lambda Nn\beta^{2} q^4}}{Nq^2+2\lambda (n+p)}.
\end{aligned}
\end{equation}
\end{lemma}
\begin{proof}
Note that 
\begin{equation} 
\begin{aligned}
\|RZ'(ZZ'+\lambda I_{n+p+1})^{-1}\|\leq \|R\|\|Z'(ZZ'+\lambda I_{n+p+1})^{-1}\|.
\end{aligned}
\end{equation}
For the term $\|R\|$, using $R_{i,j}$ to denote its $(i,j)$ entry, we have 
\begin{equation} \label{touse2}
\begin{aligned}
\|R\|\leq \|R\|_{F}=\sqrt{\sum_{i=1}^{n}\sum_{j=1}^{N}R_{i,j}^2}\leq \sqrt{Nn\beta^2 q^4},
\end{aligned}
\end{equation}
where the second inequality is due to the fact that $\|z^{i}_{0}\|_{1}=q$ for all $i=1,\ldots,N$, the assumption that $q< c$, and Assumption \ref{ass:remainder}.

For the term $\|Z'(ZZ'+\lambda I_{n+p+1})^{-1}\|$, we have
\begin{equation*} 
\begin{aligned}
&\|Z'(ZZ'+\lambda I_{n+p+1})^{-1}\|\\
&=\sqrt{\|(ZZ'+\lambda I_{n+p+1})^{-1}ZZ'(ZZ'+\lambda I_{n+p+1})^{-1}\|}.
\end{aligned}
\end{equation*}
Note that 
\begin{equation} \label{touse3}
\begin{aligned}
&\|(ZZ'+\lambda I_{n+p+1})^{-1}ZZ'(ZZ'+\lambda I_{n+p+1})^{-1}\|=\\
&\|(ZZ'+\lambda I_{n+p+1})^{-1}(ZZ'+\lambda I_{n+p+1})(ZZ'+\lambda I_{n+p+1})^{-1}\\
&\quad -\lambda(ZZ'+\lambda I_{n+p+1})^{-1}(ZZ'+\lambda I_{n+p+1})^{-1}\|\\
&\leq \|(ZZ'+\lambda I_{n+p+1})^{-1}\|+\lambda\|(ZZ'+\lambda I_{n+p+1})^{-1}\|^2.
\end{aligned}
\end{equation}
From Weyl's inequality \cite{horn2012matrix}, we have
\begin{equation*} 
\begin{aligned}
\|(ZZ'+\lambda I_{n+p+1})^{-1}\|&=\frac{1}{\lambda_{min}(ZZ'+\lambda I_{n+p+1})}\\
&\leq \frac{1}{\lambda_{min}(ZZ')+\lambda}.
\end{aligned}
\end{equation*}
Using the above inequality and \eqref{touse3}, since $N\geq 4(n+p)$ and $q\leq \sqrt{n+p}$, we can apply Lemma \ref{lemma:PE} to get
\begin{equation*} 
\begin{aligned}
&\|Z'(ZZ'+\lambda I_{n+p+1})^{-1}\|\\
&\leq \sqrt{\frac{2(n+p)}{Nq^2+2\lambda (n+p)}}+\frac{2\sqrt{\lambda} (n+p)}{Nq^2+2\lambda (n+p)},
\end{aligned}
\end{equation*}
where we used the relationship that $\sqrt{a+b}\leq \sqrt{a}+\sqrt{b}$ for $a,b\geq 0$.

Finally, combining the above inequality with \eqref{touse2}, and after some algebraic manipulations, we have the desired result.
\end{proof}

\subsection{Main Result} \label{main}
Now we present our main result, a finite sample upper bound of the system identification error \eqref{error}.
\begin{theorem} \label{thm1}
Suppose that Algorithm \ref{algo1} is used to generate data.  Let $N\geq 4(n+p)$ and $q\leq \sqrt{n+p}$. Fix constants $c$, $\beta$ that satisfy Assumption \ref{ass:remainder}, and a confidence parameter $\delta \in (0,1)$. Then if $q< c$, with probability at least $1-\delta$, the estimation error of Algorithm \ref{algo2} satisfies
\begin{equation} \label{thm1_b}
\begin{aligned}
\|\hat{\Theta}-\Theta\|&\leq \underbrace{\frac{5 \sigma_{w} \sqrt{\log\frac{9^n}{\delta}+(n+p+1)\log(1+\frac{4(n+p)}{q^2})}}{\sqrt{Nq^2/(n+p)+\lambda}}}_\text{Error due to noise}\\
&+\underbrace{\sqrt{\frac{2(n^2+np)}{1+\gamma}}\beta q}_\text{Error due to nonlinearity}\\
&+\underbrace{\frac{2(n+p)(\lambda\|\Theta\|+\sqrt{\lambda Nn\beta^2 q^4})}{2\lambda(n+p)+Nq^2}}_\text{Error due to regularization},
\end{aligned}
\end{equation}
where $\gamma=\frac{\lambda(n+p)}{Nq^2}$.
\end{theorem}
\begin{proof}
Recall the estimation error in \eqref{error}. We have
\begin{equation} 
\begin{aligned}
\|\hat{\Theta}-\Theta\|&\leq\lambda\|\Theta\|\|(ZZ'+\lambda I_{n+p+1})^{-1}\|\\
&+\|RZ'(ZZ'+\lambda I_{n+p+1})^{-1}\|\\
&+\|WZ'(ZZ'+\lambda I_{n+p+1})^{-1/2}\| \times\\
&\|(ZZ'+\lambda I_{n+p+1})^{-1/2}\|.\\
\end{aligned}
\end{equation}
Noting that 
\begin{equation} 
\begin{aligned}
\|(ZZ'+\lambda I_{n+p+1})^{-1/2}\|&=\frac{1}{\sqrt{\lambda_{min}(ZZ'+\lambda I_{n+p+1})}}\\
&\leq \frac{1}{\sqrt{\lambda_{min}(ZZ')+\lambda}},
\end{aligned}
\end{equation}
from Weyl's inequality \cite{horn2012matrix}, the result directly follows from applying Lemma \ref{lemma:PE}, Lemma \ref{bound noise}, and Lemma \ref{bound nonlinearity} after some algebraic manipulations.
\end{proof}
\begin{remark} \textbf{Interpretation of Theorem \ref{thm1}}.
Note that Theorem \ref{thm1} holds irrespective of the spectral radius of the system matrix $A$, which captures a well known advantage of the multiple trajectories setup. Below we discuss other key insights provided by Theorem \ref{thm1}.

\textbf{Trade-off between error due to noise and error due to nonlinearity:}
Suppose that $\lambda=0$ for now. When the system is perfectly linear, one has $\beta=0$. Consequently, the upper bound in Theorem \ref{thm1} only contains the error due to noise, which goes to zero with a rate of $\mathcal{O}(\frac{1}{\sqrt{N}})$. This implies a consistent estimator of which the convergence rate matches the results in the existing literature for learning perfectly linear system using random inputs \cite{dean2019sample,sarkar2019near}. When there does exist nonlinearity, i.e., $\beta >0$, one can observe that the error due to nonlinearity can be made arbitrarily small by choosing a smaller $q$ used in Algorithm \ref{algo1} (recall that $q$ captures the magnitude of the initial conditions). On the other hand, a smaller $q$ would also make the denominator of the term capturing error due to noise small, thus leading to a larger error due to noise. In other words, if one starts close enough to the reference point (by setting $q$ to be small), one would have less bias due to nonlinearity, at the cost of having a smaller signal to noise ratio (thus a larger error due to noise). However, the error due to noise can always be made almost zero by increasing the number of experiments $N$. Consequently, if one can afford to generate a large amount of data, it is preferable to use a small $q$ due to the low bias introduced by the nonlinear terms, and the small error introduced by the noise (which is due to the large amount of samples). These insights are different from system identification for truly linear systems, where it is commonly believed that a larger signal to noise ratio is always better. We will also illustrate these ideas empirically in Section \ref{exp}.

\textbf{Role of regularization:}
Suppose that $N, q$ are fixed. As $\lambda$ becomes larger, we can observe that both the error due to noise and the error due to nonlinearity goes to zero, and the error due to regularization converges to $\|\Theta\|$. This result implies that setting $\lambda$ to be relatively large can be helpful if $\sigma_{w}$ is large (system is very noisy) or $\beta$ is large (system has strong nonlinearity), while $\|\Theta\|$ is small. However, the optimal $\lambda$ can be hard to obtain if (some upper bounds of) the parameters in \eqref{thm1_b} are unknown in advance. In practice, cross validation techniques \cite{refaeilzadeh2009cross} are commonly used to select a good value of $\lambda$.
 \end{remark}

\section{Numerical Examples} \label{exp}
In this section, we provide simulated numerical examples to validate the insights for system identification using Algorithm \ref{algo1} and Algorithm \ref{algo2}. We also compare the results against the single trajectory setup, where the input is set to be independent zero mean Gaussian, with slight adjustments to deal with the offset $o$ in our setup \eqref{system}, i.e., by appending ones in the regressor matrix. More specifically, we still use Algorithm \ref{algo2} in the single trajectory setup, but the dataset is generated without restarting the system, see \cite{sarkar2019near, ye2021sample} for examples. Such comparisons are made since Gaussian inputs are commonly used in the literature on linear system identification \cite{dean2019sample,oymak2019non}. For simplicity, we set $\lambda=0$ for all experiments. All results are averaged over 10 independent experiments.
\subsection{System with mild nonlinearity}
In the first example, we investigate the performance of the system identification algorithms under mild nonlinearity. The model we use here captures the dynamics of a nonlinear pendulum.\footnote{https://courses.engr.illinois.edu/ece486/fa2019/handbook/lec02.html} The system states are the pendulum angle and its velocity, and the input is the torque applied. We set the mass and length of the pendulum to be $1$ kg and $1$ meter, respectively. After discretization using Euler's method by setting the sampling time to be 0.05 seconds, the dynamics is given by
\begin{equation} 
\begin{aligned}
\begin{bmatrix} 
x_{1,k+1}\\
x_{2,k+1}\\
\end{bmatrix}=
\begin{bmatrix} 
x_{1,k}+0.05x_{2,k}\\
-0.49\sin(x_{1,k})+x_{2,k}+0.05u_{k}\\
\end{bmatrix}+w_{k},
\end{aligned}
\end{equation}
where we set $w_{k}$ to be independent Gaussian random vectors with zero mean and covariance matrix given by $0.25 I_{2}$. The linearized system matrices around the origin are given by 
\begin{equation}
\begin{aligned}
A=
\begin{bmatrix}
1&0.05&\\
-0.49&1\\
\end{bmatrix},
B=
\begin{bmatrix}
0\\
0.05\\
\end{bmatrix},
o=
\begin{bmatrix}
0\\
0\\
\end{bmatrix}.
\end{aligned}
\end{equation}
We plot the system identification error using Algorithm \ref{algo1} and Algorithm \ref{algo2} versus the number of experiments $N$ for $q=1.2,0.9, 0.6$ in Fig.~\ref{pend_multi}. As can be observed, a smaller $q$ could lead to a larger error when $N$ is small, due to a smaller signal to noise ratio. However, a smaller $q$ can eventually result in a smaller error when $N$ is large enough due to less bias, which confirms our observations in Theorem \ref{thm1}.

In the single trajectory setup, we plot the error using i.i.d zero mean Gaussian inputs with different variance $\sigma_{u}^2$, where $N$ here represents the number of samples used in the single trajectory. A common heuristic is that one should apply small inputs to learn a good linear approximation around a given reference point, i.e., the variance $\sigma_{u}^2$ should be small. However, as shown in Fig.~\ref{pend_sing}, the error plateaus at around 0.6, even for small variance inputs. The key reason is that the random input and process noise can always drive the system states to undesired regions and excite the higher order terms, unless the input is carefully designed. In fact, the paper \cite{sarker2023accurate} shows that random inputs in the single trajectory setup could result in inconsistent estimation under certain conditions even for Lipschitz nonlinearity. 
\begin{figure}[ht]
\minipage[t]{0.43\textwidth}
\includegraphics[width=\linewidth]{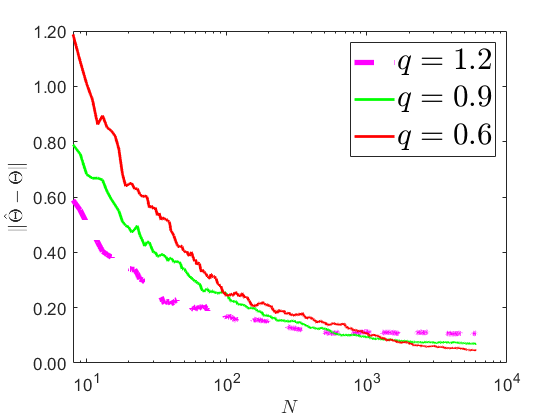}
\caption{System identification error using Algorithms \ref{algo1}-\ref{algo2} with different $q$, mild nonlinearity}
\label{pend_multi}
\endminipage \hfill
\end{figure}

\begin{figure}[ht] 
\minipage[t]{0.43\textwidth}
\includegraphics[width=\linewidth]{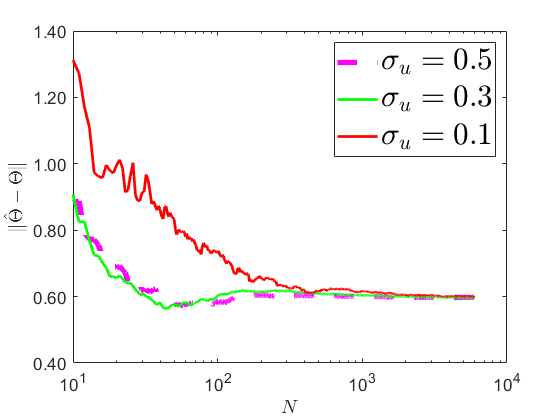}
\caption{System identification error using a single trajectory with different $\sigma_{u}$, mild nonlinearity}
\label{pend_sing}
\endminipage \hfill
\end{figure}

\subsection{System with strong nonlinearity}
In the second example, we investigate the performance of the system identification algorithms under strong nonlinearity (where the assumption of lipschitzness used in \cite{sarker2023accurate} no longer holds). The virtual model we use here is given by

\begin{equation} 
\begin{aligned}
\begin{bmatrix} 
x_{1,k+1}\\
x_{2,k+1}\\
\end{bmatrix}
&=\begin{bmatrix} 
0.9 & 0.5 \\
0 & 0.8 \\
\end{bmatrix}
\begin{bmatrix} 
x_{1,k}\\
x_{2,k}\\
\end{bmatrix}
+\begin{bmatrix} 
1\\
1\\
\end{bmatrix}u_{k}\\
&+\begin{bmatrix} 
x_{1,k}^3+x_{2,k}^5\\
x_{1,k}x_{2,k}\\
\end{bmatrix}
+\begin{bmatrix} 
1\\
1\\
\end{bmatrix}
+w_{k},
\end{aligned}
\end{equation}
where we again set $w_{k}$ to be independent Gaussian random vectors with zero mean and covariance matrix given by $0.25 I_{2}$. 

Again, we plot the system identification error using Algorithm \ref{algo1} and Algorithm \ref{algo2} versus the number of experiments $N$ for $q=0.6, 0.4,0.2$. As can be observed, similar trends still hold, i.e., a smaller $q$ results in a larger error when $N$ is small, but is beneficial  in the long run, even for system with relatively strong nonlinearity. 

In contrast, in the single trajectory setup, we applied i.i.d zero mean Gaussian inputs with variance $\sigma_{u}^2=0.1^2, 0.01^2, 0.001^2$. However, all of them fail to converge and result in numerical issues since the noise and non-zero offset drive the system states to regions where nonlinearity dominates. 
\section{Conclusion and future work} \label{sec: conclusion}
In this paper, we proposed system identification algorithms to learn the linearized model of a system. Unlike existing works, we assume that the underlying dynamics could be nonlinear. We presented a finite sample error bound of the algorithms, which shows that one can learn the linearized dynamics with arbitrarily small error given sufficiently many samples, and demonstrates a trade-off between the error due to noise and the error due to nonlinearity. Our bound further characterizes the benefits of using regularization. As shown in \cite{ahmadi2021safely}, initializing states at different locations might come at different costs. Consequently, future work would focus on studying how to optimize the data collection procedure under constraints on initial state/input.

\begin{figure}[ht] 
\minipage[t]{0.43\textwidth}
\includegraphics[width=\linewidth]{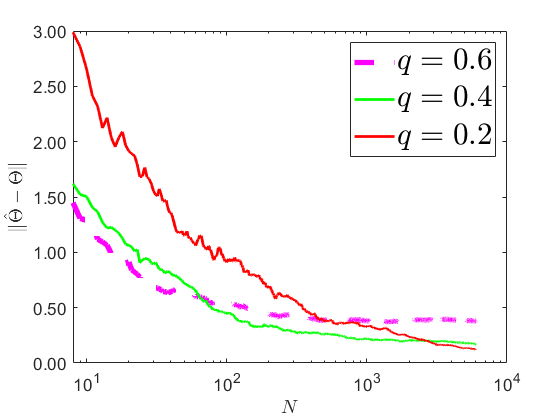}
\caption{System identification error using Algorithms \ref{algo1}-\ref{algo2} with different $q$, strong nonlinearity}
\label{strong_multi}
\endminipage \hfill
\end{figure}



\bibliographystyle{IEEEtran}
\bibliography{main}

\begin{thebibliography}{10}
\providecommand{\url}[1]{#1}
\csname url@samestyle\endcsname
\providecommand{\newblock}{\relax}
\providecommand{\bibinfo}[2]{#2}
\providecommand{\BIBentrySTDinterwordspacing}{\spaceskip=0pt\relax}
\providecommand{\BIBentryALTinterwordstretchfactor}{4}
\providecommand{\BIBentryALTinterwordspacing}{\spaceskip=\fontdimen2\font plus
\BIBentryALTinterwordstretchfactor\fontdimen3\font minus
  \fontdimen4\font\relax}
\providecommand{\BIBforeignlanguage}[2]{{%
\expandafter\ifx\csname l@#1\endcsname\relax
\typeout{** WARNING: IEEEtran.bst: No hyphenation pattern has been}%
\typeout{** loaded for the language `#1'. Using the pattern for}%
\typeout{** the default language instead.}%
\else
\language=\csname l@#1\endcsname
\fi
#2}}
\providecommand{\BIBdecl}{\relax}
\BIBdecl

\bibitem{athey2018impact}
S.~Athey, ``The impact of machine learning on economics,'' in \emph{The
  economics of artificial intelligence: An agenda}.\hskip 1em plus 0.5em minus
  0.4em\relax University of Chicago Press, 2018, pp. 507--547.

\bibitem{mitchell2007machine}
T.~M. Mitchell \emph{et~al.}, \emph{Machine learning}.\hskip 1em plus 0.5em
  minus 0.4em\relax McGraw-hill New York, 2007, vol.~1.

\bibitem{ljung1999system}
L.~Ljung, ``System identification,'' \emph{Wiley encyclopedia of electrical and
  electronics engineering}, pp. 1--19, 1999.

\bibitem{rugh1996linear}
W.~J. Rugh, \emph{Linear system theory}.\hskip 1em plus 0.5em minus 0.4em\relax
  Prentice-Hall, Inc., 1996.

\bibitem{bauer1999consistency}
D.~Bauer, M.~Deistler, and W.~Scherrer, ``Consistency and asymptotic normality
  of some subspace algorithms for systems without observed inputs,''
  \emph{Automatica}, vol.~35, no.~7, pp. 1243--1254, 1999.

\bibitem{jansson1998consistency}
M.~Jansson and B.~Wahlberg, ``On consistency of subspace methods for system
  identification,'' \emph{Automatica}, vol.~34, no.~12, pp. 1507--1519, 1998.

\bibitem{dean2019sample}
S.~Dean, H.~Mania, N.~Matni, B.~Recht, and S.~Tu, ``On the sample complexity of
  the linear quadratic regulator,'' \emph{Foundations of Computational
  Mathematics}, pp. 1--47, 2019.

\bibitem{fattahi2018data}
S.~Fattahi and S.~Sojoudi, ``Data-driven sparse system identification,'' in
  \emph{Proc. Allerton Conference on Communication, Control, and Computing},
  2018, pp. 462--469.

\bibitem{zheng2020non}
Y.~Zheng and N.~Li, ``Non-asymptotic identification of linear dynamical systems
  using multiple trajectories,'' \emph{IEEE Control Systems Letters}, vol.~5,
  no.~5, pp. 1693--1698, 2020.

\bibitem{xin2022learning}
L.~Xin, G.~Chiu, and S.~Sundaram, ``Learning the dynamics of autonomous linear
  systems from multiple trajectories,'' in \emph{2022 American Control
  Conference (ACC)}.\hskip 1em plus 0.5em minus 0.4em\relax IEEE, 2022, pp.
  3955--3960.

\bibitem{simchowitz2018learning}
M.~Simchowitz, H.~Mania, S.~Tu, M.~I. Jordan, and B.~Recht, ``Learning without
  mixing: Towards a sharp analysis of linear system identification,'' in
  \emph{Proc. Conference On Learning Theory}, 2018, pp. 439--473.

\bibitem{oymak2019non}
S.~Oymak and N.~Ozay, ``Non-asymptotic identification of {LTI} systems from a
  single trajectory,'' in \emph{American control conference}.\hskip 1em plus
  0.5em minus 0.4em\relax IEEE, 2019, pp. 5655--5661.

\bibitem{simchowitz2019learning}
M.~Simchowitz, R.~Boczar, and B.~Recht, ``Learning linear dynamical systems
  with semi-parametric least squares,'' in \emph{Proc. Conference on Learning
  Theory}, 2019, pp. 2714--2802.

\bibitem{sarkar2019nonparametric}
T.~Sarkar, A.~Rakhlin, and M.~A. Dahleh, ``Nonparametric finite time {LTI}
  system identification,'' \emph{arXiv preprint arXiv:1902.01848}, 2019.

\bibitem{faradonbeh2018finite}
M.~K.~S. Faradonbeh, A.~Tewari, and G.~Michailidis, ``Finite time
  identification in unstable linear systems,'' \emph{Automatica}, vol.~96, pp.
  342--353, 2018.

\bibitem{sarkar2019near}
T.~Sarkar and A.~Rakhlin, ``Near optimal finite time identification of
  arbitrary linear dynamical systems,'' in \emph{Proc. International Conference
  on Machine Learning}, 2019, pp. 5610--5618.

\bibitem{sarker2023accurate}
A.~Sarker, P.~Fisher, J.~E. Gaudio, and A.~M. Annaswamy, ``Accurate parameter
  estimation for safety-critical systems with unmodeled dynamics,''
  \emph{Artificial Intelligence}, p. 103857, 2023.

\bibitem{sattar2022non}
Y.~Sattar and S.~Oymak, ``Non-asymptotic and accurate learning of nonlinear
  dynamical systems,'' \emph{Journal of Machine Learning Research}, vol.~23,
  no. 140, pp. 1--49, 2022.

\bibitem{mania2020active}
H.~Mania, M.~I. Jordan, and B.~Recht, ``Active learning for nonlinear system
  identification with guarantees,'' \emph{arXiv preprint arXiv:2006.10277},
  2020.

\bibitem{foster2020learning}
D.~Foster, T.~Sarkar, and A.~Rakhlin, ``Learning nonlinear dynamical systems
  from a single trajectory,'' in \emph{Learning for Dynamics and
  Control}.\hskip 1em plus 0.5em minus 0.4em\relax PMLR, 2020, pp. 851--861.

\bibitem{mauroy2016linear}
A.~Mauroy and J.~Goncalves, ``Linear identification of nonlinear systems: A
  lifting technique based on the koopman operator,'' in \emph{2016 IEEE 55th
  Conference on Decision and Control (CDC)}.\hskip 1em plus 0.5em minus
  0.4em\relax IEEE, 2016, pp. 6500--6505.

\bibitem{hao2022deep}
W.~Hao, B.~Huang, W.~Pan, D.~Wu, and S.~Mou, ``Deep koopman representation of
  nonlinear time varying systems,'' \emph{arXiv preprint arXiv:2210.06272},
  2022.

\bibitem{cobzacs2019lipschitz}
{\c{S}}.~Cobza{\c{s}}, R.~Miculescu, A.~Nicolae \emph{et~al.}, \emph{Lipschitz
  functions}.\hskip 1em plus 0.5em minus 0.4em\relax Springer, 2019.

\bibitem{rivasplata2012subgaussian}
O.~Rivasplata, ``Subgaussian random variables: An expository note,''
  \emph{Internet publication, PDF}, vol.~5, 2012.

\bibitem{courant1965introduction}
R.~Courant, F.~John, A.~A. Blank, and A.~Solomon, \emph{Introduction to
  calculus and analysis}.\hskip 1em plus 0.5em minus 0.4em\relax Springer,
  1965, vol.~1.

\bibitem{folland2005higher}
G.~B. Folland, ``Higher-order derivatives and taylor’s formula in several
  variables,'' \emph{Preprint}, pp. 1--4, 2005.

\bibitem{hoerl1970ridge}
A.~E. Hoerl and R.~W. Kennard, ``Ridge regression: Biased estimation for
  nonorthogonal problems,'' \emph{Technometrics}, vol.~12, no.~1, pp. 55--67,
  1970.

\bibitem{abbasi2011improved}
Y.~Abbasi-Yadkori, D.~P{\'a}l, and C.~Szepesv{\'a}ri, ``Improved algorithms for
  linear stochastic bandits,'' \emph{Advances in neural information processing
  systems}, vol.~24, 2011.

\bibitem{xin2023learning}
L.~Xin, L.~Ye, G.~Chiu, and S.~Sundaram, ``Learning dynamical systems by
  leveraging data from similar systems,'' \emph{arXiv preprint
  arXiv:2302.04344}, 2023.

\bibitem{horn2012matrix}
R.~A. Horn and C.~R. Johnson, \emph{Matrix analysis}.\hskip 1em plus 0.5em
  minus 0.4em\relax Cambridge university press, 2012.

\bibitem{refaeilzadeh2009cross}
P.~Refaeilzadeh, L.~Tang, and H.~Liu, ``Cross-validation.'' \emph{Encyclopedia
  of database systems}, vol.~5, pp. 532--538, 2009.

\bibitem{ye2021sample}
L.~Ye, H.~Zhu, and V.~Gupta, ``On the sample complexity of decentralized linear
  quadratic regulator with partially nested information structure,'' \emph{IEEE
  Transactions on Automatic Control}, 2022.

\bibitem{ahmadi2021safely}
A.~A. Ahmadi, A.~Chaudhry, V.~Sindhwani, and S.~Tu, ``Safely learning dynamical
  systems from short trajectories,'' in \emph{Learning for Dynamics and
  Control}.\hskip 1em plus 0.5em minus 0.4em\relax PMLR, 2021, pp. 498--509.

\end{thebibliography}
\end{document}